\documentclass{birkjour}
 \newtheorem{thm}{Theorem}[section]
 
 \newtheorem{lem}[thm]{Lemma}
 
 \theoremstyle{definition}
 
 \theoremstyle{remark}
 
 \newtheorem*{ex}{Example}
 \numberwithin{equation}{section}

\begin{document}

%-------------------------------------------------------------------------
% editorial commands: to be inserted by the editorial office
%
%\firstpage{1} \volume{228} \Copyrightyear{2004} \DOI{003-0001}
%\seriesextraline{This is the Concrete Title of this Book\br H.E. R and S.T.C. W, Eds.}
%\firstpage{1}

\title[A Condition To Non-Existence Of Incompressible Velocity Fields.]
 {A Condition On Spherical Surfaces To Non-Existence Of Incompressible Velocity Fields.}

%----------Author 1
\author[Manuel Garc\'ia-Casado]{Manuel Garc\'ia-Casado}

\address{%
ICP-CSIC\\
Marie Curie, 2\\
Cantoblanco (Madrid) E-28049\\
Spain}

\email{mgac@icp.csic.es}

%\thanks{This work was completed with the support of our
%----------Author 2
%\author{A Second Author}
%\address{The address of\br
%the second author\br
%sitting somewhere\br
%in the world}
%\email{dont@know.who.knows}
%\subjclass{Primary 99Z99; Secondary 00A00}

%\date{January 1, 2004}

\begin{abstract}
In an incompressible velocity field, the surface area of a volume varies with time, but volume remains unchanged. If incidentally the surface becomes spherical along time, the area reaches a local minimum, since sphere has the least area that surrounds a volume. So the area is a function of time that is locally convex at this point. When applied to an incompressible Navier--Stokes fluid, this property is used to compute an inequality that suggest a criterion to non-existence of initial configurations of velocity fields, revealing its impossibility to evolve with time. Three velocity fields are proposed as examples. One of them agrees the inequality, the other two violates it.
\end{abstract}

%%% ----------------------------------------------------------------------
\maketitle
%%% ----------------------------------------------------------------------
%\tableofcontents
\section{Introduction}
\label{intro}
Every dynamical system, described by differential equations, deals with the initial value problem. This is, given an initial condition, one tries to determine whether the system can evolve with time or not beginning from that condition. Sometimes, it may be possible to determine whether there are one or several solutions for equations with the initial condition. There was proposed in \cite{RefJ1} that initial conditions for incompressible Navier-Stokes velocity fields are useful to find its time evolution, in such a way that given suitable restrictions to the initial velocity field, the system is determined at least for any finite time after. In the same way, Beale, Kato and Majda \cite{RefJ2} proved that a smooth velocity field may lose its regularity some time after, in such a way that the maximum vorticity becomes unbounded. Hence, to find properties of the initial velocity field is a challenge. In this paper we propose a criterion to the non-existence of some of these velocity fields.
\section{Transport theorem for surfaces}
Reylods transport theorem \cite{RefB3} is a very useful tool since it lets introduce the time derivative of a volume integral inside the integrand of a static volume integral. We would like to do the same with a surface integral. It is, to transform the time derivative of a surface integral, which surface is moving and changing its shape, and to obtain a fixed surface integral with a time derivative inside its integrand. The next theorem shows how to find this issue (see \cite{RefB4}).
\begin{thm}
Let $\vec{u}\left(t,\vec{x}\right)\in \mathbf{R}^{3}$ be a velocity field with components $u_{i}$ that are enough smooth, and let be $f\left(t,\vec{x}\right)\in\mathbf{R}$ also a smooth function. Let $\Omega\subset\mathbf{R}^{3}$ be a region of the field with boundary $\partial\Omega$. The unitary normal vector to $\partial\Omega$ is $\vec{n}$,with components $n_{i}$. Then,
\begin{equation}
\frac{d}{dt}\int_{\partial\Omega\left(t\right)}f d^{2}x=\int_{\partial\Omega}\left[\partial_{t}f +u_{i}\partial_{i}f+\left(\partial_{i}u_{i}-\epsilon_{ij}n_{i}n_{j}\right)f\right]d^{2}x,
\label{eq theorem 1 001} 
\end{equation}
where $\epsilon_{ij}=\frac{1}{2}\left(\partial_{i}u_{j}+\partial_{j}u_{i}\right)$ is the infinitesimal strain tensor (defined by, e. g., \cite{RefB5}).
\end{thm}
\begin{proof}
Since the vector normal to the surface is unitary and the surface closed, we can use Gauss theorem
\begin{eqnarray}
\frac{d}{dt}\int_{\partial\Omega\left(t\right)}f d^{2}x=\frac{d}{dt}\int_{\partial\Omega\left(t\right)}n_{i} f n_{i}d^{2}x\nonumber\\
=\frac{d}{dt}\int_{\Omega\left(t\right)}\partial_{i}\left(f n_{i}\right)d^{3}x.
\label{eq theorem 1 002} 
\end{eqnarray}
Then, we can apply Reynolds transport theorem to the volume integral,  
\begin{equation}
\frac{d}{dt}\int_{\partial\Omega\left(t\right)}f d^{2}x=\int_{\Omega}\left[\partial_{t}\partial_{i}\left(f n_{i}\right)+u_{j}\partial_{j}\partial_{i}\left(f n_{i}\right)+\partial_{j}u_{j}\partial_{i}\left(f n_{i}\right)\right]d^{3}x.\label{eq theorem 1 003} 
\end{equation}
Using the chain rule twice in the second term in the integral of right hand side and taking in to account that time and space derivatives commutes in the first term of right hand side,  
\begin{eqnarray}
\frac{d}{dt}\int_{\partial\Omega\left(t\right)}f d^{2}x=\int_{\Omega}\left\{\partial_{i}\left[\left(\partial_{t}+u_{j}\partial_{j}\right)\left(f n_{i}\right)\right]
-\partial_{j}\left(fn_{i}\partial_{i}u_{j}\right)\right\}d^{3}x+\nonumber\\
+\int_{\Omega}\partial_{i}\left(f n_{i}\partial_{j}u_{j}\right)d^{3}x.\label{eq theorem 1 004} 
\end{eqnarray}
So, we can use again Gauss theorem,
\begin{eqnarray}
\frac{d}{dt}\int_{\partial\Omega\left(t\right)}fd^{2}x=\int_{\partial\Omega}f\left(\partial_{t}+u_{j}\partial_{j}\right)\left(\frac{1}{2}n_{i}n_{i}\right)d^{2}x\nonumber\\ 
+\int_{\partial\Omega}\left\{n_{i}n_{i}\left[\left(\partial_{t}+u_{j}\partial_{j}\right)f+\partial_{j}u_{j}f\right]
-f\partial_{i}u_{j}n_{j}n_{i}\right\}d^{2}x.
\label{eq theorem 1 005} 
\end{eqnarray}
The first term of right hand side is the derivative of a constant and it vanishes. Then, from
\begin{eqnarray}
\frac{d}{dt}\int_{\partial\Omega\left(t\right)}f d^{2}x=\int_{\partial\Omega}\left[\left(\partial_{t}+u_{i}\partial_{i}\right)f
+\left(\partial_{i}u_{i}-\partial_{i}u_{j}n_{j}n_{i}\right)f\right]d^{2}x,
\label{eq theorem 1 006} 
\end{eqnarray}
the relation (\ref{eq theorem 1 001}) arises and the theorem is proved.
\end{proof}
Equation (\ref{eq theorem 1 001}) is similar to the transport theorem for moving surfaces \cite{RefJ6}-\cite{RefJ7}, which usually is written in terms of both, normal velocity and curvature of the surface. Now that we know the rate of change of the surface integral of a magnitude with time, we would like to know whether the area of the surface grows, diminishes or remains constant with time when the volume does not change. We  knows a particular case yet. One of the properties of the sphere is that it has the least area that encloses a volume. So, the area of the sphere only can increase or be the same few time after. This means that the area is a convex function of time near of the minimum. The next theorem refletcs this situation.
\begin{thm}
Let $\vec{u}\left(t,\vec{x}\right)\in \mathbf{R}^{3}$ be a velocity field with components $u_{r},u_{\theta},u_{\phi}$ in spherical coordinates. Let $\mathbf{S}^{3}\subset\mathbf{R}^{3}$ be a spherical region of the field with boundary $\mathbf{S}^{2}$ and radium $r$. Also, there exists only one region $\Omega\left(t\right)\subset \mathbf{R}^{3}$ for $t \neq t_{0}$ such as $\Omega\left(t\right)\rightarrow \mathbf{S}^{3}$ when $t\rightarrow t_{0}$. For each $\mathbf{S}^{3}$, and every $t$, if the velocity field holds the incompressibility statement, $\vec{\nabla}\cdot\vec{u}=0$, then
\begin{equation}
\int^{\pi}_{0}\int^{2 \pi}_{0}\left[\epsilon^{2}_{rr}-\frac{D}{Dt}\epsilon_{rr} \right]r^{2}\sin \theta d\theta d\phi\geq 0,
\label{eq theorem 2 001} 
\end{equation}
where $\epsilon_{rr}=\partial_{r}u_{r}$ ( for stain tensor in spherical coordinates see \cite{RefB8}).
\end{thm}
\begin{proof}
 Taking into account the very well known isoperimetric inequality for three dimensions \cite{RefB9}-\cite{RefJ10}, we have
\begin{equation}
\int_{\partial\Omega\left(t\right)}d^{2}x \geq 3 \left(\frac{4}{3}\pi \right)^{\frac{1}{3}}\left[\int_{\Omega\left(t\right)}d^{3}x\right]^{\frac{2}{3}},
\label{eq theorem 2 002}
\end{equation}
where the equality holds for the sphere $\mathbf{S}^{3}$.
We substract the area of $\mathbf{S}^{2}$ on both sides,
\begin{eqnarray}
\int_{\partial\Omega\left(t\right)}d^{2}x-\int_{\mathbf{S}^{3}}d^{2}x \geq 3 \left(\frac{4}{3}\pi \right)^{\frac{1}{3}}\left[\int_{\Omega\left(t\right)}d^{3}x\right]^{\frac{2}{3}}-\int_{\mathbf{S}^{2}}d^{2}x\nonumber\\
\geq
3 \left(\frac{4}{3}\pi \right)^{\frac{1}{3}}\left\{\left[\int_{\Omega\left(t\right)}d^{3}x\right]^{\frac{2}{3}}-\left[\int_{\mathbf{S}^{3}}d^{3}x\right]^{\frac{2}{3}}\right\}.
\label{eq theorem 2 003}
\end{eqnarray}
Due to the incompressibility of the fluid, $\mathbf{S}^{3}$ and $\Omega$ have the same volume. The right hand side of (\ref{eq theorem 2 003}) then vanishes
\begin{equation}
\int_{\partial\Omega\left(t\right)}d^{2}x-\int_{\mathbf{S}^{2}}d^{2}x \geq 0.
\label{eq theorem 2 004}
\end{equation}
In addition, the area time derivative is given by (\ref{eq theorem 1 001}), with $f= 1$ and $\partial_{i}u_{i}= 0$,
\begin{eqnarray}
 \left[\frac{d}{dt}\int_{\partial\Omega\left(t\right)}d^2 x\right]\left(t_0\right)=-\int^{\pi}_{0}\int^{2\pi}_{0}\partial_{r}u_{r} r^{2} \sin \theta d\theta d\phi\nonumber\\=-\partial_{r}\left[\int^{\pi}_{0}\int^{2\pi}_{0}u_{r} r^{2} \sin \theta d\theta d\phi\right]+\frac{2}{r}\int^{\pi}_{0}\int^{2\pi}_{0}u_{r} r^{2} \sin \theta d\theta d\phi\nonumber\\=-\partial_{r}\left[\int_{\mathbf{S}^3}\partial_{i}u_{i} d^3x\right]+\frac{2}{r}\int_{\mathbf{S}^3}\partial_{i}u_{i} d^3x=0.
\label{eq theorem 2 005} 
\end{eqnarray}
So the area of a sphere reaches its minimum at time $t=t_0$ in a incompressible velocity field. This property together with (\ref{eq theorem 2 004}) means that the area is a local convex function of time in a range close to $t_{0}$.
Therefore, the second time derivative of this function at $t_{0}$ holds
\begin{equation}
 \left[\frac{d^{2}}{dt^{2}}\int_{\partial\Omega\left(t\right)}d^2 x\right]\left(t_0\right) \geq 0.
\label{eq theorem 2 006}
\end{equation}
The second time derivative of the area can be computed applying (\ref{eq theorem 1 001}) twice
\begin{eqnarray}
\left[\frac{d^{2}}{dt^{2}}\int_{\partial\Omega\left(t\right)}d^2 x\right]\left(t_0\right)=\left[\frac{d}{dt}\int_{\partial\Omega\left(t\right)}\left(-n_{i}n_{j}\partial_{i}u_{j}\right)d^2 x\right]\left(t_0\right) \nonumber\\
=\left[\int_{\partial\Omega\left(t\right)}\left\{\left(n_{i}n_{j}\partial_{i}u_{j}\right)^{2}-\frac{D}{Dt}\left(n_{i}n_{j}\partial_{i}u_{j}\right)\right\}d^2 x\right]\left(t_0\right)  \nonumber\\
=\int_{\mathbf{S}^{2}}\left\{\left(\partial_{r}u_{r}\right)^2-\frac{D}{Dt}\left(\partial_{r}u_{r}\right)\right\}d^2x,
\label{eq theorem 2 007}
\end{eqnarray}
where in the last line we have used that the normal vector to the surface of the sphere only has radial component. Taken (\ref{eq theorem 2 006}) together with (\ref{eq theorem 2 007}), we can find (\ref{eq theorem 2 001}) at time $t=t_0$. But the spherical surface is independent of time, since the time dependency is in the integrand of the last line of (\ref{eq theorem 2 007}). This is, at every time, for every spherical surface, there exist a volume, which is a function of time, that converges to the sphere. Then (\ref{eq theorem 2 001}) is held at every instant of time.
\end{proof}
 Given that we have a surface integral, it does not matter what is the velocity distribution inside the sphere but just that velocity distribution on its surface. Therefore, this theorem asserts that if there exist at least a sphere in the domain of the incompressible velocity field that violates (\ref{eq theorem 2 001}), evolution with time is forbidden for that velocity field. The next lemma applies this theorem to incompressible Navier-Stokes fluids. 
\begin{lem}
Let $\vec{u}\left(t,\vec{x}\right)\in \mathbf{R}^{3}$ be an incompressible velocity field, $\vec{\nabla}\cdot \vec{u}=0$, which evolves in time according to the Navier-Stokes equations 
\begin{equation}
\partial_{t}\vec{u}+\vec{u}\cdot\vec{\nabla}\vec{u}=\mu \triangle\vec{u}-\vec{\nabla} p.
\label{eq lemma 1 001} 
\end{equation}
Here, $p$ is the pressure, the density is $\rho=1$ and $\mu$ is the viscosity. Velocity components in spherical coordinates are denoted by $u_{r},u_{\theta},u_{\phi}$ ($\theta$ and $\phi$ are polar and azimuth angles, respectively). Then, at every time $t$, for every spherical region of the field $\mathbf{S}^{3}\subset\mathbf{R}^{3}$ with boundary $\mathbf{S}^{2}$ and radius $r$, we have
\begin{equation}
\int_{\mathbf{S}^{2}}\left\{\partial^{2}_{r}p(t,\vec{x})-F(t,\vec{x})-\mu G(t,\vec{x})+\left(\partial_{r}u_{r}(t,\vec{x})\right)^{2}\right\}d^{2} x \geq 0
\label{eq lemma 1 002} 
\end{equation}
where
\begin{eqnarray}
F(t,r,\theta,\phi)&=&\partial_{r}\left(\frac{u^{2}_{\theta}+u^{2}_{\phi}}{r}\right)-\left(\partial_{r}u_{r}\right)^{2}-\partial_{r}\left(\frac{u_{\theta}}{r}\right)\partial_{\theta}u_{r}\nonumber\\&&
-\partial_{r}\left(\frac{u_{\phi}}{r}\right)\frac{\partial_{\phi}u_{r}}{\sin \theta},\label{eq lemma 1 003}\\
G(t,r,\theta,\phi)&=&\partial^{3}_{r}u_{r}+2\partial^{2}_{r}\left(\frac{u_{r}}{r}\right) +\frac{1}{r^{2}}\left(\partial^2_{\theta}\partial_{r}u_{r}+\cot \theta\partial_{\theta}\partial_{r}u_{r}\right)\nonumber\\&&-\frac{2}{r^3}\left(\partial^2_{\theta}u_{r}+\cot \theta\partial_{\theta}u_{r}\right).
\label{eq lemma 1 004}
\end{eqnarray}
\end{lem}
\begin{proof}
The radial direction of the equation (\ref{eq lemma 1 001}) is given by 
\begin{eqnarray}
\partial_{t}u_{r}+u_{r}\partial_{r}u_{r}+\frac{u_{\theta}}{r}\partial_{\theta}u_{r}+\frac{u_{\phi}}{r \sin \theta}\partial_{\phi}u_{r}-\frac{u^2_\theta+u^2_\phi}{r}= \nonumber\\-\partial_{r} p+ \mu \left[\partial^{2}_{r}u_{r}+2\partial_{r}\left(\frac{u_{r}}{r}\right)+\frac{1}{r^2}\partial^{2}_{\theta}u_{r}+\frac{\cot \theta}{r^{2}}\partial_{\theta}u_{r}\right]
\label{eq lemma 1 005}
\end{eqnarray}
(see, e.g., \cite{RefB11})Next we can take the partial derivative of this relation with respect to $r$,  and then we use the identity
\begin{equation}
\frac{DQ}{Dt}=\partial_{t}Q+u_{r}\partial_{r}Q+\frac{u_{\theta}}{r}\partial_{\theta}Q+\frac{u_{\phi}}{r \sin \theta}\partial_{\phi}Q
\label{eq lemma 1 006}
\end{equation}
(where $Q$ is a scalar magnitude) to group terms, obtaining
\begin{equation}
\frac{D}{Dt}\left(\partial_{r}u_{r}\right) =-\partial^{2}_{r}p+F+\mu G,
\label{eq lemma 1 007}
\end{equation}
Substitution of this relation on (\ref{eq theorem 2 001}) gives rise to (\ref{eq lemma 1 002}).
\end{proof}
This lemma means that if we find at least a spherical surface for which the incompressible velocity field does not hold (\ref{eq lemma 1 002}), that field can not evolve acording to Navier--Stokes equations. Notice that the lemma is only useful when the inequality is violated. Lets see it with three examples.
\begin{ex}[1]
At time $t_0$, let a velocity field be given by
\begin{equation}
\left\{
\begin{array}{l l}
u_{r}=0 \\
u_{\theta}=0 \\
u_{\phi}=r \sin \theta \\ 
\end{array}
\right.
\label{eq examp 1 001} 
\end{equation}
(in spherical coordinates) inside a bigger sphere of ratio $R$.
The fluid of this velocity field spins around the $z$ axis and is divergent-free. We would like to confirm that  (\ref{eq lemma 1 002}) is correct. The computation of (\ref{eq lemma 1 003}) and (\ref{eq lemma 1 004})  to this velocity  field gives
\begin{eqnarray}
F\left(r,\theta,\phi \right)&=&\sin ^2 \theta 
\label{eq examp 1 002} \\
G\left(r,\theta,\phi \right) &=&0
\label{eq examp 1 003} 
\end{eqnarray}
Computation of double radium derivative of pressure is more difficult. We can work out the pressure, as usual, by solving the Poisson equation obtained when we take the divergence of incompressible Navier-Stokes equations  (\ref{eq lemma 1 001}). So, this non-local function of spatial derivatives of velocity is
\begin{equation}
p\left(t,\vec{x}'\right)=\frac{1}{4\pi}\int_{\mathbf{R}^3}\frac{\vec{\nabla} \cdot \left(\vec{u}\cdot \vec{\nabla} \vec{u}\right)}{\left|\vec{x}-\vec{x}'\right| } d^3 x.
\label{eq examp 1 004} 
\end{equation} 
In our case, the corresponding derivations and integration in the sphere of radium $R$ gives us
\begin{equation}
p\left(r,\theta,\phi\right)=\frac{1}{3}r^2-R^2,
\label{eq examp 1 005} 
\end{equation}
being $R \geq r$.
Then, the double time derivative of the surface area that converges to a sphere of radium $r$ in this velocity field is
\begin{eqnarray}
\int^{\pi}_{0}\int^{2 \pi}_{0}\left\{\partial^{2}_{r}p-F-\mu G+\left(\partial_{r}u_{r}\right)^{2}\right\} r^{2}\sin \theta d\theta d\phi \nonumber\\
=\int^{\pi}_{0}\int^{2 \pi}_{0}\left\{\frac{2}{3}-\sin^2 \theta \right\} r^{2}\sin \theta d\theta d\phi=0.
\label{eq examp 1 006} 
\end{eqnarray}
This time, the result does agree with the inequality (\ref{eq lemma 1 002}). But still remains regions of the $\mathbf{R}^3$ where we can look for violation of the inequality.
Now, the system of reference is shifted a distance $L$ from the $z$ axis along $y$ axis (see Fig. \ref{fig:1} ), instead of be on it. With the identities given by
\begin{equation}
\left\{\begin{array}{rcl}
&&r\cos \theta = r' \cos \theta' \\
&&r\sin \theta \cos \phi = r' \sin \theta' \cos \phi'  \\
&&L=r' \sin \theta' \sin \phi' - r \sin \theta \sin \phi
\end{array}
\right.
\label{eq examp 1 007} 
\end{equation}
\begin{figure}
\centering
\resizebox{0.75\textwidth}{!}{
  \includegraphics{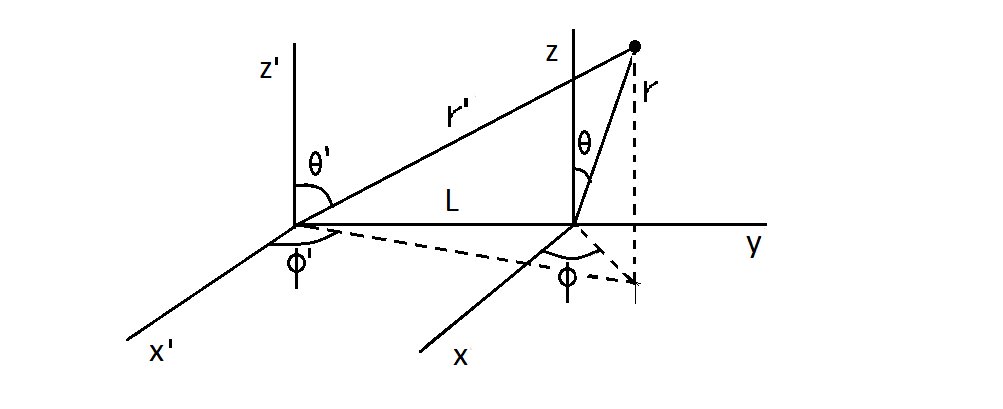}
}
\caption{Coordinates origin shifted a distance $L$ along $y$ axis.}
\label{fig:1}    
\end{figure}
we can change
\begin{equation}
\left\{\begin{array}{rcl}
&& u_{r'}=0 \\
&& u_{\theta'}=0  \\
&& u_{\phi'}=r' \sin \theta' \\
&& p= \frac{1}{3}r'^{2}-R^{2}
\end{array}
\right.
\label{eq examp 1 008}
\end{equation}
by
\begin{equation}
\left\{\begin{array}{rcl}
&& u_{r}=L \sin \theta \cos \phi \\
&& u_{\theta}=L \cos \theta \cos \phi  \\
&& u_{\phi}=r \sin \theta - L \sin \phi \\
&& p= \frac{1}{3}\left(r^{2}+L^{2}+2L r \sin \theta \sin \phi\right)-R^{2}
\end{array}
\right.
\label{eq examp 1 009}
\end{equation}
with $R \geq L+r$. Of course, this velocity field still has $\vec{\nabla}\cdot\vec{u}=0$. Repeating again the steps like before, we find
\begin{eqnarray}
F\left(r,\theta,\phi \right)&=&\sin^{2} \theta+ \frac{L}{r}\sin \theta \sin \phi-\frac{L^2}{r^2} \left(\sin^{2} \phi +\cos \theta\cos \phi \sin \phi\right) 
\label{eq examp 1 010}  \\
G\left(r,\theta,\phi \right) &=& \frac{L\cos \phi}{r^{3} \sin \theta}\left( 8 \sin^{2} \theta - \cos^{2} \theta\right)
\label{eq examp 1 011} 
\end{eqnarray}
and hence,
\begin{equation}
\int^{\pi}_{0}\int^{2 \pi}_{0}\left\{\partial^{2}_{r}p-F-\mu G+\left(\partial_{r}u_{r}\right)^{2}\right\} r^{2}\sin \theta d\theta d\phi= 2 \pi L^{2}.
\label{eq examp 1 012} 
\end{equation}
This result also agrees with the inequality (\ref{eq lemma 1 002}). Moreover, it is independent of the radium $R$. So, this inequality can be extrapolated to $\mathbf{R}^3$ doing $R \to \infty$ or $L \to \infty$ and using revolution symmetry around $z'$ axis.
\end{ex}
\begin{ex}[2]
In this example, we will see that the divergent-free velocity field (see Fig. \ref{fig:2}) given in spherical coordinates, at a time $t_0$ ,by
\begin{equation}
\left\{\begin{array}{rcl}
&& u_{r}=0 \\
&& u_{\theta}=0  \\
&& u_{\phi}=r^{k} \sin \theta,
\label{eq examp 2 001} 
\end{array}
\right.
\end{equation}
with $2 \leq k \in \mathbf{N}$, does not hold the inequality (\ref{eq lemma 1 002}) for every sphere that is inside a bigger sphere of radium $R$. As before, first we compute $F$, $G$ from velocity and its derivatives, and then, $p$ from them and from the integral over the sphere of radium $R$ (being $R\geq r$). So, we obtain
\begin{eqnarray}
F\left(r,\theta,\phi \right)&=&\left(2k-1\right)r^{2k-2}\sin^{2} \theta \label{eq examp 2 002}\\
G\left(r,\theta,\phi \right) &=& 0 \label{eq examp 2 003}\\
p\left(r,\theta,\phi \right) &=&- \sum^{\infty}_{n \in\mathbf{N}^0\backslash\left\{2k\right\}}\left(\frac{R^{2k-n}}{r^{2k-n}}-\frac{2n+1}{2k+n+1}\right)\frac{r^{2k}}{2k-n}\nonumber\\
&&\int^{\pi}_{0}\left[\left(k-1\right)\sin^2 \theta' +1\right]\sin \theta' P_{n}\left(\cos \left(\theta-\theta'\right)\right)d \theta' \nonumber\\
&&-\left(\frac{1}{2k+1}-\ln\frac{R}{r}\right)r^{2k}\nonumber\\
&&\int^{\pi}_{0}\left[\left(k-1\right)\sin^2 \theta' +1\right]\sin \theta' P_{2k}\left(\cos \left(\theta-\theta'\right)\right)d \theta' \nonumber\\
\label{eq examp 2 003} 
\end{eqnarray}
(where $\mathbf{N}^0\backslash\left\{2k\right\}=0,1,2...2k-1,2k+1,...$ and $P_n\left(x\right)$ are Legrendre polimominals) and hence,
\begin{eqnarray}
\int^{\pi}_{0}\int^{2 \pi}_{0}\left\{\partial^{2}_{r}p-F-\mu G+\left(\partial_{r}u_{r}\right)^{2}\right\} r^{2}\sin \theta d\theta d\phi= \nonumber\\
-2\pi r^{2k}\sum^{\infty}_{n \in\mathbf{N}^0\backslash\left\{2k\right\}}\left[\frac{2k\left(2k-1\right)}{2k+n+1}-\frac{2k\left(2k-1\right)}{2k-n}+\frac{n\left(n-1\right)}{2k-n}\frac{R^{2k-n}}{r^{2k-n}}\right]A_n\nonumber\\
-2\pi r^{2k}\left[\frac{2k\left(4k-1\right)-1}{4k+1}-2k\left(2k-1\right)\ln\frac{R}{r}\right]A_{2k}-\left(2k-1\right)\frac{8}{3}\pi r^{2k}\nonumber\\
\label{eq examp 2 004}
\end{eqnarray}
where
\begin{equation}
A_n=\int^{\pi}_{0}\int^{\pi}_{0}\left(\sin^2 \theta +1\right)\sin \theta \sin \theta' P_n \left(\cos \left( \theta - \theta'\right)\right) d \theta d \theta'
\label{eq examp 2 005}
\end{equation}
Notice that, since $-2\leq A_n \leq 2$ and $r \leq R$, the summatory converges. Moreover, we have used $R$ as a parameter to compute the pressure and we can make it very large. When $R \gg r$, we can approach (\ref{eq examp 2 004}) by
\begin{equation}
-2\pi r^{2k}\frac{1}{k-1}\frac{R^{2k-2}}{r^{2k-2}}A_2 + O\left(\frac{R^{2k-3}}{r^{2k-3}}\right).
\label{eq examp 2 006}
\end{equation}
However, since $A_2=12/5$, it is impossible that it holds the inequality
\begin{equation}
-2\pi r^{2k}\frac{1}{k-1}\frac{R^{2k-2}}{r^{2k-2}}A_2 + O\left(\frac{R^{2k-3}}{r^{2k-3}}\right) \geq 0.
\label{eq examp 2 007}
\end{equation}
So we conclude that surprisingly the velocity field (\ref{eq examp 2 001}) can not evolve according to incompressible Navier-Stokes equations.
\begin{figure}
\centering
\resizebox{0.5\textwidth}{!}{
  \includegraphics{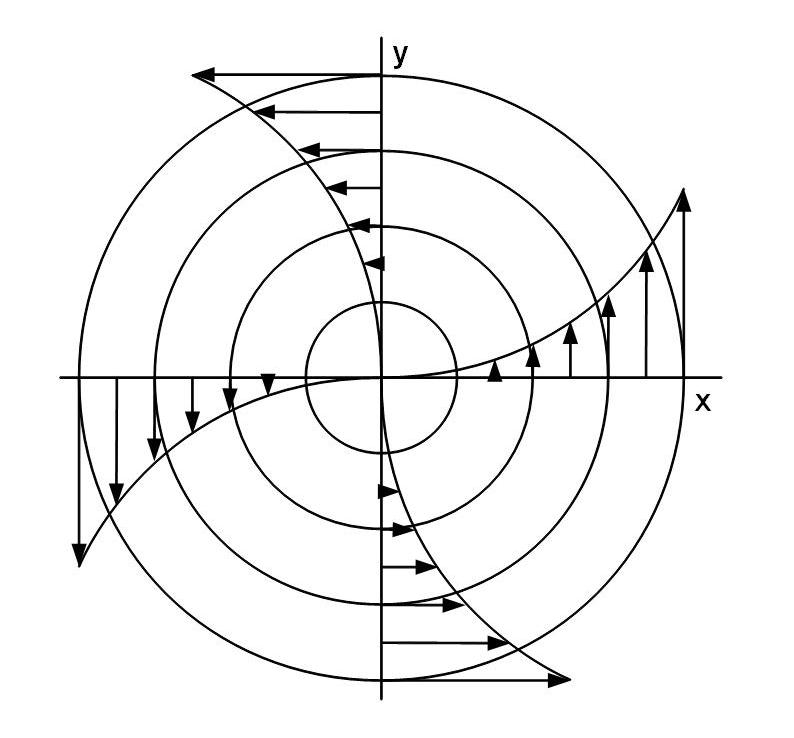}
}
\caption{Velocity profile of (\ref{eq examp 2 001}) in the x-y plane.}
\label{fig:2}    
\end{figure}
\end{ex}
\begin{ex}[3]
In this last example, we will see that the divergent-free velocity field (see Fig. \ref{fig:3}) given, at a time $t_0$ ,by
\begin{equation}
\left\{\begin{array}{rcl}
&& u_{r}=\left(R^{2}-r^{2}\sin^{2}\theta\right)\cos\theta \\
&& u_{\theta}=-\left(R^{2}-r^{2}\sin^{2}\theta\right)\sin\theta   \\
&& u_{\phi}=0,
\label{eq examp 3 001} 
\end{array}
\right.
\end{equation}
with $0 \leq r < \infty$ and $0\leq \theta\leq \pi$, does not hold the inequality (\ref{eq lemma 1 002}) for every sphere.  Proceeding as before, we again compute $F$, $G$ and $p$, this one worked out throught an integral over all three dimensional space. So, we obtain
\begin{eqnarray}
F\left(r,\theta,\phi \right)&=&\left(r^{2}\sin^{2}\theta-R^{2}\right)\left(r^{2}\sin^{2}\theta-\cos^{2}\theta\right)\sin^{2} \theta, \label{eq examp 3 002}\\
G\left(r,\theta,\phi \right) &=& \left(\frac{10 R ^{2}}{r^{2}}-2\sin^{2}\theta\right)\frac{\cos\theta}{r} ,\label{eq examp 3 003}\\
p\left(r,\theta,\phi \right) &=&0, \nonumber\\
\label{eq examp 3 003} 
\end{eqnarray}
 and hence,
\begin{eqnarray}
\int^{\pi}_{0}\int^{2 \pi}_{0}\left\{\partial^{2}_{r}p-F-\mu G+\left(\partial_{r}u_{r}\right)^{2}\right\} r^{2}\sin \theta d\theta d\phi= \frac{16 \pi r^{2}}{5}\left(R^{2}-\frac{4}{7} r^{2}\right)\nonumber\\
\label{eq examp 3 004}
\end{eqnarray}
However, the inequality (\ref{eq lemma 1 002}) does not hold when the radius of the probe sphere is
\begin{equation}
r > \frac{\sqrt{7}}{2}R.
\label{eq examp 3 007}
\end{equation}
Then,  the velocity field (\ref{eq examp 3 001}) can not evolve according to incompressible Navier-Stokes equations.
\begin{figure}
\centering
\resizebox{0.5\textwidth}{!}{
  \includegraphics{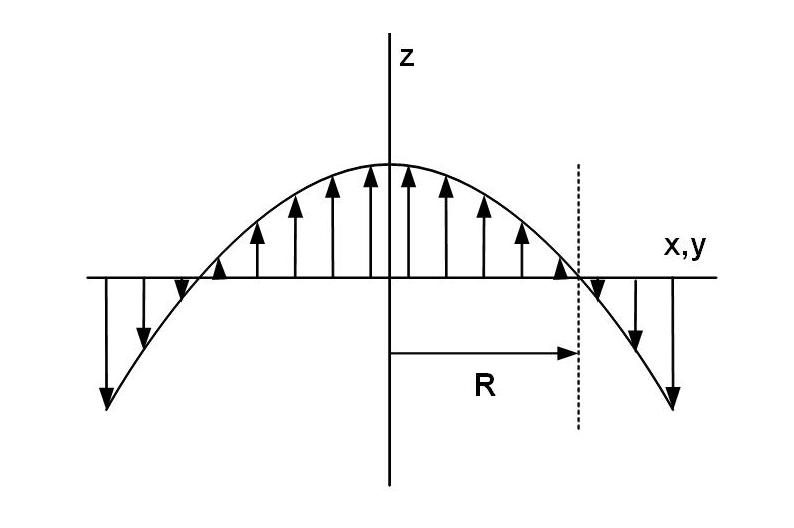}
}
\caption{Velocity profile of (\ref{eq examp 3 001}) in a plane that contains the z axis.}
\label{fig:3}    
\end{figure}
\end{ex}
\section{Conclusion}
We have shown that given an incompressible velocity field at a initial time, we can test whether its time evolution is forbidden by a criterion. It is related with the non-negativeness of the double time derivative of the area of a volume that becomes a sphere at that instant. Of course, if velocity agrees the inequality at a give time it also agrees that the reminder of the time, because in other case, the field had not evolved to reach that instant. In particular we have worked the inequality out to a Navier--Stokes fluids. We also have found two particular incompressible velocity fields that can not evolve acording to Navier--Stokes equations.

%%\begin{figure}% Use the relevant command for your figure-insertion program
%%\centering
%%\resizebox{0.75\textwidth}{!}{%
%%  \includegraphics{leer.eps}
%\vspace{5cm}       % Give the correct figure height in cm
%%\caption{Please write your figure caption here}
%%\label{fig:1}       % Give a unique label
%%\end{figure}
%%\begin{table}
%%\caption{Please write your table caption here}
%%\label{tab:1}       % Give a unique label
%%\centering
%%\begin{tabular}{lll}
%%\hline\noalign{\smallskip}
%%first & second & third  \\
%%\noalign{\smallskip}\hline\noalign{\smallskip}
%%number & number & number \\
%%number & number & number \\
%%\noalign{\smallskip}\hline
%%\end{tabular}
% Or use
%\vspace*{5cm}  % with the correct table height
%%\end{table}\alpha\alpha
%
% BibTeX users please use
% \bibliographystyle{}
% \bibliography{}

\begin{thebibliography}{}
\bibitem{RefJ1}
Beale, J.T.; Comm. Pure Appl. Math., \textbf{34} (1981), pp.359-392.
\bibitem{RefJ2}
Beale, J.T., Kato, T., Majda, A.; Commun. Math. Phys., \textbf{94} (1984), pp.61-66.
\bibitem{RefB3}
Reynolds, O.  \textit{Papers on mechanical and physical subjets} (Vol. 3, The Sub-Mechanicals of The Universe, Cambridge University Press, 1903) pp.14
\bibitem{RefB4}
Batchelor, G. K., \textit{An Introduction To Fluid Dynamics} (Cambridge University Press, 2000) pp.132
\bibitem{RefB5}
Malvern, L.E., \textit{Introduction to the Mechanics of a Continuous Medium} (Prentice-Hall, 1969) pp.129-138
\bibitem{RefJ6}
Gurtin, M.E., Struthers, A., Williams, V.O.; Q. Appl. Math. \textbf{47} (1989), pp.773-777.
\bibitem{RefJ7}
Jaric, J.P.; Int. J. Engng. Sci. \textbf{30} (1992), pp.1535-1542.
\bibitem{RefB8}
Sadd, M.H., \textit{Elasticity. Theory, Applications and Numerics, ($2^{nd}$ Edition)} (Academic Press, 2009) pp.46-47
\bibitem{RefB9}
Chavel, I., \textit{Isoperimetric Inequalities. Differential Geometric And Analitic Perspectives} (Cambridge Tracts in Mathematics (no.145), 2001) pp.2-3
\bibitem{RefJ10}
Fusco, N., Maggi, F., Pratelli, A.; Ann. of Math., \textbf{168} (2008), pp.941-980.
\bibitem{RefB11}
Durst, F., \textit{Fluid Mechanics: An Introduction to the Theory of Fluid Flows} (Springer-Verlag, 2008) pp. 139
\end{thebibliography}

% end of file template.tex

\end{document}